\documentclass[a4,12pt]{article}
\usepackage{amsmath,amssymb,enumerate,bbm,theorem}
\usepackage[english]{babel}

\usepackage[margin=1in]{geometry}

\newcommand{\assign}{:=}
\newcommand{\nequiv}{\not\equiv}
\newcommand{\nobracket}{}
\newcommand{\nocomma}{}
\newcommand{\tmem}[1]{{\em #1\/}}
\newcommand{\tmop}[1]{\ensuremath{\operatorname{#1}}}

\newenvironment{enumerateroman}{\begin{enumerate}[i.] }{\end{enumerate}}
\newenvironment{proof}{\noindent\textbf{Proof\ }}{\hspace*{\fill}$\Box$\medskip}
\newtheorem{definition}{Definition}
{\theorembodyfont{\rmfamily}\newtheorem{example}{Example}}
\newtheorem{lemma}{Lemma}
\newtheorem{proposition}{Proposition}
{\theorembodyfont{\rmfamily}\newtheorem{remark}{Remark}}

\begin{document}

\title{On the equation $\nabla \phi + \phi X = 0$ and its relation to
Schr{\"o}dinger ground states}
\author{Kurt Pagani\\pagani@scios.ch\\To the memory of Joseph Hersch 1925-2012}

\maketitle

\abstract{We present some simple relations between the absolute
minimizers of the functional $\vert\vert\nabla\phi+\phi
X\vert\vert$, where $X$ is a vector field on ${\mathbb R}^n$, and ground state 
solutions to the (non-relativistic) Schr{\"o}dinger equation. 
This article is a byproduct of the study of the more general functional 
$\vert\vert A\,Du+F'(u) X \vert\vert$.}

\section{Ground states}

In the following let $X, Y, \ldots$denote vector fields on $\mathbbm{R}^n$,
and $\langle f, g \rangle$ and $\| f \|$ the $L^2 ( \mathbbm{R}^n)$ scalar
product and norm respectively. We will use the same notation for vector fields
as well, that is \ $\langle X, Y \rangle = \sum \langle X_i, Y_i \rangle$ and
$\| X \| = \sqrt{\langle X, X \rangle}$. All vector fields and functions shall
be real valued. For any $V \in L^1_{\tmop{loc}} ( \mathbbm{R}^n)$ we define
the Schr{\"o}dinger ground state energy as
\begin{equation}
  E_0 ( V) = \inf \left\{ \| \nabla \phi \|^2 + \int_{\mathbbm{R}^n} V ( x) |
  \phi ( x) |^2 d x : \, \phi \in C_0^{\infty} ( \mathbbm{R}^n),
  \, \| \phi \| = 1 \right\}, \label{schrgs}
\end{equation}
where $C_0^{\infty} ( \mathbbm{R}^n)$ means, as usual, the space of smooth
functions having compact support. Note that we use units such that the
(non-relativistic, time independent) Schr{\"o}dinger equation has the form
\begin{equation}
  - \Delta \psi + V \psi = E \psi \hspace{2em} \tmop{in} \hspace{1em}
  \mathcal{D}' ( \mathbbm{R}^n) . \label{schreq}
\end{equation}

The leading actor in this article is the non-coercive functional
\begin{equation}
  \mathcal{J}_X ( \phi) = \int_{\mathbbm{R}^n} | \nabla \phi + \phi X |^2 d x
  = \| \nabla \phi + \phi X \|^2,
\end{equation}
which is well defined on $C_0^{\infty} ( \mathbbm{R}^n)$ for any locally
square integrable vector field $X.$
\newpage
\begin{definition}
  A function $\phi_0$ is called a ground state to $\mathcal{J}_X ( \phi)$, if
  the following conditions are satisfied:
  \begin{enumerateroman}
    \item $\phi_0 \nequiv 0$,
    
    \item $\mathcal{J}_X ( \phi_0) = 0$,
    
    \item $\| \nabla \phi_0 \| < \infty .$
  \end{enumerateroman}
\end{definition}

Thus, when we set
\begin{equation}
  \Lambda ( X) = \inf \left\{ \mathcal{J}_X ( \phi) : \, \phi \in
  C_0^{\infty} ( \mathbbm{R}^n), \, \| \phi \| = 1 \right\},
\end{equation}
a necessary condition for a ground state to exist is $\Lambda ( X) = 0.$
Formally, the Euler equation to the minimum problem above is easily calculated
to be
\begin{equation}
  - \Delta u + ( | X |^2 - \tmop{div} X) u = \Lambda ( X) u, \label{eulereq}
\end{equation}
revealing the connection with the Schr{\"o}digner equation $\left( \ref{schreq}
\right) .$ On the other hand, if there is a ground state to $\mathcal{J}_X$,
say $\phi_0,$ it has to satisfy the first order equation (a.e.)
\begin{equation}
  \nabla \phi_0 ( x) + \phi_0 ( x) X ( x) = 0 \label{maineq}
\end{equation}
since $\mathcal{J}_X ( \phi_0) = \| \nabla \phi_0 + \phi_0 X \| = 0$. Indeed,
applying the divergence operator to $\left( \ref{maineq} \right)$ and using the latter
again in the resulting expression, gives $\left( \ref{eulereq} \right)$ with
$\Lambda ( X) = 0$, as expected. It is clear that the equation $\left(
\ref{maineq} \right)$ is generally easier to solve than $\left( \ref{eulereq}
\right)$.

Now let us state a simple proposition:

\begin{proposition}
  Let $X \in L^2_{\tmop{loc}}$ such that $\tmop{div} ( X) \in L_{\tmop{loc}}^1
  ( \mathbbm{R}^n)$, then for all $\phi \in C_0^{\infty} ( \mathbbm{R}^n)$ the
  following assertions are true:
  \begin{equation}
    \| \nabla \phi + X \phi \|^2 = \| \nabla \phi \|^2 + \| X \phi \|^2 -
    \int_{\mathbbm{R}^n} \tmop{div} ( X)  | \phi ( x) |^2 d x \label{peq1}
  \end{equation}
  and
  \begin{equation}
    \| \nabla \phi \|^2  \| X \phi \|^2 \geqslant \frac{1}{4} \left(
    \int_{\mathbbm{R}^n} \tmop{div} ( X)  | \phi ( x) |^2 d x \right)^2
    \label{peq2}
  \end{equation}
  with equality \ if and only if
  \begin{equation}
    \| \nabla \phi_0 + X \phi_0 \|^{} = 0 \label{peq3}
  \end{equation}
  for some functions $\phi_0$ such that $\| \nabla \phi_0 \| < \infty .$ For
  any such function $\phi_0$ then
  \begin{equation}
    \| \nabla \phi_0 \|^2 = \| X \phi_0 \|^2 = \frac{1}{2}
    \int_{\mathbbm{R}^n} \tmop{div} ( X)  | \phi_0 ( x) |^2 d x \label{peq4}
  \end{equation}
  holds.
\end{proposition}

\begin{proof}
  When expanding $\| \nabla \phi + X \phi \|^2$ one gets $\| \nabla \phi \|^2
  + \| X \phi \|^2 + 2 \langle \nabla \phi, X \phi \rangle$, where the scalar
  product $2 \langle \nabla \phi, X \phi \rangle $may be rewritten to $\langle
  \nabla | \phi |^2, X \rangle$. Since $\tmop{div} ( X) \in L^1_{\tmop{loc}} (
  \mathbbm{R}^n)$ by definition, the identity $\left( \ref{peq1} \right)$
  follows by application of the divergence theorem. Inequality $\left(
  \ref{peq2} \right)$ and $\left( \ref{peq3} \right)$ are immediate
  consequences of the discriminant condition applied to $\| \nabla \phi +
  \lambda X \phi \|^2 = \| \nabla \phi \|^2 + \lambda^2 \| X \phi \|^2 - 2
  \lambda \langle \phi, \phi \tmop{div} ( X) \rangle \geqslant 0$ for all
  $\lambda \in \mathbbm{R}$. Finally, $\left( \ref{peq4} \right)$ follows from
  $\left( \ref{peq3} \right)$ and $\left( \ref{peq1}, \ref{peq2} \right)$.
\end{proof}

When we define
\[ V_{X, \lambda} = | X |^2 - \tmop{div} ( X) + \lambda, \]
for $X \in L^{2_{}}_{\tmop{loc}}$ , $\tmop{div} ( X) \in L^1_{\tmop{loc}}$ ,
$\lambda \in \mathbbm{R} \nocomma,$ then $\left( \ref{peq1} \right)\, $reads
\[ \| \nabla \phi + X \phi \|^2 = \int_{\mathbbm{R}^n} | \nabla \phi |^2 d x +
   \int_{\mathbbm{R}^n} ( V_{X, \lambda} ( x) - \lambda) | \phi ( x) |^2 d x,
\]
and recalling $\left( \ref{schrgs} \right)$, we get
\[ \inf_{\| \phi \| = 1} \| \nabla \phi + X \phi \|^2 = E_0 ( V_{X, \lambda})
   - \lambda \geqslant 0, \]
which motivates the next proposition:

\begin{proposition}
  For $V \in L_{\tmop{loc}}^1 ( \mathbbm{R}^n)$ and $E_0 ( V)$ as defined in
  $\left( \ref{schrgs} \right)$ it holds
  \begin{equation}
    E_0 ( V) \geqslant \sup_X  \left\{ \Lambda ( X) + \inf_{\| \phi \| = 1}
    \int_{\mathbbm{R}^n} ( V ( x) - | X |^2 + \tmop{div} X)  | \phi ( x) |^2 d
    x \right\} . \label{jhineq}
  \end{equation}
  and
  \begin{equation}
    E_0 ( V) \leqslant \inf_X  \left\{ \Lambda ( X) + \sup_{\| \phi \| = 1}
    \int_{\mathbbm{R}^n} ( V ( x) - | X |^2 + \tmop{div} X)  | \phi ( x) |^2 d
    x \right\} . \label{jhineq2}
  \end{equation}
  \[  \]
\end{proposition}

If $X$ and $V$ are pointwise defined (e.g. continuous) then there is
a simple lower bound to the Schr{\"o}dinger ground state energy:
\[ E_0 ( V) \geqslant \inf_{\mathbbm{R}^n} \{ V ( x) - | X |^2 + \tmop{div} X
   \}, \]
valid for any reasonable vector field. The lower bound above was
was orally communicated to us by the late Jospeh Hersch many years
ago and, actually, we found a reference to it in one of his numerous papers 
\cite{jh1} . 
Note that $\left(
\ref{jhineq} \right)$ may be simplified by taking the supremum over ground 
states only, because then $\Lambda ( X) = 0.$ \\

\begin{proof}
  Setting $B_{V, X} ( \phi) = \int_{\mathbbm{R}^n} ( V ( x) - | X |^2 +
  \tmop{div} X)  | \phi ( x) |^2 d x$, the upper bound
  \[ E_0 ( V) \leqslant \| \nabla \phi + \phi X \|^2 + B_{V, X} ( \phi) \]
  follows easily from $\left( \ref{schrgs} \right)$ and
  $\left(\ref{peq1}\right)$. Adding $\int V | \phi |^2 d x$ to
  $\left(\ref{peq1}\right)$ and rearranging, yields
  \[ \| \nabla \phi + \phi X \|^2 + B_{V, X} ( \phi) = \| \nabla \phi \|^2 +
     \int_{\mathbbm{R}^n} V ( x) | \phi ( x) |^2 d x. \]
  Taking the infimum over \{$\| \phi \| = 1$\} on both sides gives
  \[ \Lambda ( X) + \inf_{\| \phi \| = 1} B_{V, X} ( \phi) \leqslant E_0 ( V)
     . \]
  Since the right hand side is independent of $X$, taking the supremum over
  $X$ proves $\left( \ref{jhineq} \right) .$ Now $\left( \ref{jhineq2}
  \right)$ follows from the upper bound: $\inf_{\| \phi \| = 1} \{ E_0 ( V) -
  B_{V, X} ( \phi) \} \leqslant \Lambda ( X)$. Indeed, $E_0 ( V) + \inf_{\|
  \phi \| = 1} ( - B_{V, X} ( \phi)) = E_0 ( V) - \sup_{\| \phi \| = 1} B_{V,
  X} ( \phi) \leqslant \Lambda ( X) .$ 
\end{proof}

\begin{remark}
  It is intuitively obvious that not every $X$ gives rise to a ground state, or
  in other terms, the equation $\left( \ref{maineq} \right)$ may have no
  (nontrivial) solutions at all. Indeed, as will be seen shortly, $X$ must be a gradient
  (one would say {\tmem{exact}} in terms of differential forms). It is well
  known that Schr{\"o}dinger ground states may be chosen positive, so that $X
  = - \nabla \log \psi_0$ is an admissible vector field satisfying
  $\mathcal{J}_X ( \psi_0) = 0$.
  
  Noting that the linear functional
  \[ T_X ( \phi) = \int_{\mathbbm{R}^n} \langle \nabla \phi + \phi X, X
     \rangle d x \]
  satisfies
  \[ | T_X ( \phi) | \leqslant \| X \|_{2, K}  \| \nabla \phi \|_{2, K} + \|
     X \|_{2, K}^2 \| \phi \|_{\infty} \leqslant C_K \sup_{| \alpha |
     \leqslant 1}  \| \partial^{\alpha} \phi \|_{\infty}, \]
  for any $X \in L_{\tmop{loc}}^2$, thus $T \in \mathcal{D}' (
  \mathbbm{R}^n)$, so that Propostion 4 may be extended to more general
  potentials (e.g. measures) along the same lines.
\end{remark}

The ground states to $\mathcal{J}_X$ have some nice properties.

\begin{proposition}
  Let $\phi_0, \phi_1$ ground states to $\mathcal{J}_X, \mathcal{J}_Y$
  respectively, then $\phi_0 \cdot \phi_1 $is a ground state to
  $\mathcal{J}_{X + Y}$.
\end{proposition}

\begin{proof}
  $\| \nabla ( \phi_0 \phi_1) + ( X + Y) ( \phi_0 \phi_1) \| = \| \phi_0 (
  \nabla \phi_1 + Y \phi_1) + \phi_1 ( ( \nabla \phi_0 + X \phi_0)) \| = 0$
  because both terms $\nabla \phi_1 + Y \phi_1 = \nabla \phi_0 + X \phi_0 = 0$
  in $L^2 ( \mathbbm{R}^n)$ by supposition.
\end{proof}

\begin{proposition}
  Suppose $\phi_0$ is a ground state to $\mathcal{J}_X$, and let $P$ be a
  harmonic, homogeneous polynomial, satisfying \ $2 \nabla P ( x) \cdot X ( x)
  + W ( x) P ( x) = 0$, then
  \[ \phi_P ( x) = P ( x) \phi_0 ( x) \]
  is a solution of $- \Delta \phi_P + ( W + | X |^2 - \tmop{div} X) \phi_P =
  0$, in $\mathcal{D}' ( \mathbbm{R}^n)$.
\end{proposition}

\begin{proof}
  Let $\varphi \in C_0^{\infty} ( \mathbbm{R}^n) \nocomma$, then
  \[ \langle \Delta \varphi, P \phi_0 \rangle = - \langle \nabla \varphi,
     \phi_0 \nabla P + P \nabla \phi_0 \rangle  \]
  Now, since
  \[ \nabla \phi_0 + \phi_0 X = 0 \]
  by supposition, it follows $\langle \Delta \varphi, \nobracket P \phi_0)
  \rangle = - \langle \nabla \varphi, \phi_0 \nabla P - P \phi_0 X \rangle$,
  furthermore (using \ $\Delta P = 0$), $\langle \Delta \varphi, \nobracket P
  \phi_0) \rangle = - \langle \varphi, \nabla \phi_0 \nabla P - \phi_0 \nabla
  P X - P \nabla \phi_0 X - \phi_0 P \tmop{div} X \rangle,$using $\nabla
  \phi_0 + \phi_0 X = 0$ again (two times)
  \[ \langle \Delta \varphi, P \phi_0 \rangle = \langle \varphi, - \phi_0 X
     \nabla P - \phi_0 \nabla P X + P \phi_0  | X |^2 - \phi_0 P \tmop{div} X
     \rangle . \]
  Then, writing $\phi_P = P \phi_0$ and using the supposition $2 \nabla P X =
  W P$
  \[ \int_{\mathbbm{R}^n} \Delta \varphi \phi_P d x = \int_{\mathbbm{R}^n}
     \varphi ( W \phi_P + \phi_P  | X |^2 - \phi_P \tmop{div} X) d x. \]
\end{proof}

Now, let us look at some examples.

\begin{example}
  Let $X ( x) = \alpha \frac{x}{| x |^p}$, $\alpha, p \in \mathbbm{R} \nocomma
  \nocomma$. Noticing that $X$ is a gradient,
  \[ X = \nabla u = \frac{\alpha}{2 - p} \nabla | x |^{2 - p} \]
  the equation
  \[ \nabla \phi(x) + \phi(x)\, \nabla u(x) = 0 \]
  is easily solved:
  
		\[ \nabla(\log \phi + u)=0 \Rightarrow \phi(x) = C e^{-u(x)}, 
	\]
	
  thus we get
  \begin{equation}
    \phi_p ( x) = C \exp \left[ \frac{\alpha}{p - 2} | x |^{2 - p} \right]
    \label{eigenfun}
  \end{equation}
  where $C = C ( n, p, \alpha)$ is a normalization constant. A straightforward
  computation gives
  \[ | X |^2 = \alpha^2 | x |^{2 - 2 p} \in L^1_{\tmop{loc}} ( \mathbbm{R}^n)
     \ldots \tmop{if} \ldots p < 1 + \frac{n}{2}, \]
  \[ {\tmop{div} X = \alpha \frac{n}{| x |^p} - \alpha p \frac{1}{| x
     |^p} = \alpha \frac{n - p}{| x |^p} \in L^1_{\tmop{loc}} ( \mathbbm{R}^n)
     \ldots \tmop{if} \ldots p < 1 + n}, \]
  hence
  \[ V_{X, \lambda} ( x) - \lambda = \frac{\alpha^2}{| x |^{2 ( p - 1)}} -
     \alpha \frac{n - p}{| x |^p} \in L^1_{\tmop{loc}} ( \mathbbm{R}^n) \ldots
     \tmop{if} \ldots p < 1 + \frac{n}{2} . \]
  Inserting into equations $\left( \ref{peq1} \right), \left( \ref{peq2}
  \right)$ gives
  \begin{equation}
    \| \nabla \phi + X \phi \|^2 = \int_{\mathbbm{R}^n} | \nabla \phi_{} |^2 d
    x + \int_{\mathbbm{R}^n} \left( \frac{\alpha^2}{| x |^{2 ( p - 1)}} -
    \alpha \frac{n - p}{| x |^p} \right) | \phi ( x) |^2 d x \geqslant 0
    \label{exeq1}
  \end{equation}
  and
  \begin{equation}
    \int_{\mathbbm{R}^n} | \nabla \phi_{} |^2 d x \int_{\mathbbm{R}^n} | x
    |^{2 - 2 p}  | \phi ( x) |^2 d x \geqslant \frac{( n - p)^2}{4} \left(
    \int_{\mathbbm{R}^n} \frac{| \phi ( x) |^2}{| x |^p} d x \right)^2,
    \label{exeq2}
  \end{equation}
  where the equality sign holds if $\phi = \phi_p$ (provided that, of course,
  $\phi_p$ satisfies the conditions of a ground state. For instance 
	$n \geq 2,\, \alpha >
  0 \nocomma,\, p < 1 + n / 2$.
\end{example}

Let us have a closer look to the cases $p = 0, 1, 2, \ldots
\nocomma,$ revealing some old friends:

\begin{example}
  Uncertainty, harmonic oscillator $X = \alpha x$. Setting $p = 0$ in $\left(
  \ref{exeq1} \right)$ and $\left( \ref{exeq2} \right)$ yields
  \[ \| \nabla \phi + X \phi \|^2 = \int_{\mathbbm{R}^n} | \nabla \phi_{} |^2
     d x + \alpha^2 \int_{\mathbbm{R}^n} | x |^2 | \phi ( x) |^2 d x - 
		\alpha n \int
     | \phi ( x) |^2 d x, \]
  which gives the Schr{\"o}dinger ground state eigenvalue $E_0 = \alpha n$
  with eigenfunction $(\ref{eigenfun})$
  \[ \phi_0 ( x) = C e^{- \frac{\alpha}{2} | x |^2}, \]
  and the well known Heisenberg uncertainty relation:

  \[ \int_{\mathbbm{R}^n} | \nabla \phi ( x) |^2 d x \int_{\mathbbm{R}^n} | x
     |^2  | \phi ( x) |^2 d x \geqslant \frac{n^2}{4} \left(
     \int_{\mathbbm{R}^n} | \phi ( x) |^2 d x \right)^2 \label{eq11} \]
  Therefore, $\phi_0 $ is a nice function, it is in 
  $\mathcal{S} ( \mathbbm{R}^n)$
  and real analytic (like $X$). Equation $\left( \ref{eulereq} \right)$ goes
  to
  \begin{equation}
    - \Delta \phi_0 ( x) + \alpha^2 | x |^2 \phi_0 ( x) = n \alpha
    \phi_0 ( x)
  \end{equation}
  which is Schr\"odinger's equation for the harmonic oscillator.
  
  Proposition 6 shows moreover that $E_k = \alpha ( n + 2 k),$ $k = 0, 1, 2,
  \ldots$ are the higher eigenvalues with eigenfunctions
  \begin{equation}
    P_k ( x) \exp \left( - \frac{\alpha}{2} | x |^2 \right),
    \label{efosc}
  \end{equation}
  where $P_k$ is a harmonic, homogeneous polynomial of degree $k$.
\end{example}

\begin{example}
  The one electron atom: $X = \alpha \frac{x}{| x |}$. Setting $p = 1$, we get
  analogously
  \[ \| \nabla \phi + X \phi \|^2 = \| \nabla \phi \|_{}^2 +
     \int_{\mathbbm{R}^n} \left( \alpha^2 - \frac{\alpha \left( n
     - 1 \right)}{| x |} \right) | \phi ( x) |^2 d x, \]
  and
  \[ \int_{\mathbbm{R}^n} | \nabla \phi ( x) |^2 d x \int_{\mathbbm{R}^n} |
     \phi ( x) |^2 d x \geqslant \frac{\left( n^{} - 1 \right)^2}{4}
     \left( \int_{\mathbbm{R}^n} \frac{| \phi ( x) |^2}{| x |} d x \right)^2 .
  \]
  Now, $\left( \ref{eulereq} \right)$ goes to
  \[ - \Delta \phi_0 - \frac{\alpha \left( n - 1
     \right)}{| x |} \phi_0 = - \alpha^2 \phi_0 \]
  which is, when setting $n = 3 \nocomma, \, 2 \alpha = Z \nocomma$,
  the Schr{\"o}dinger equation of an electron in the field of a nucleus of 
  charge $Z$.
  The eigenfunction $(\ref{eigenfun})$ is
  \[ \phi_0 = C \exp \left( - \frac{Z}{2} | x | \right) \]
  and the corresponnding eigenvalue
  \[ E_0 = - \frac{Z^2}{4} \]
  so that 
  \[ \int_{\mathbbm{R}^3} | \nabla \phi ( x) |^2 d x - Z
     \int_{\mathbbm{R}^3}  \frac{| \phi ( x) |^2}{| x |} \geqslant -
     \frac{Z^2}{4} \int | \phi ( x) |^2 d x\]
	with equality for the $\phi_0$ above.
	
  This $\phi_0$ is still a nice function, but it is not in $\mathcal{S} (
  \mathbbm{R}^n)$ and fails to be continuously differentiable at the origin,
  i.e. $\phi_0 \in C^{\omega} ( \mathbbm{R}^n \backslash \{0\} \nobracket)$ 
  only.
  Again, Proposition 6 provides higher eigenvalues and eigenfunctions:
  \[  E_k^{} = - \left( \frac{2}{n - 1 + 2 k} \right)^2
        \alpha^2, \, k = 0, 1, 2, \ldots \] 
  \begin{equation}
    \phi_k^{} ( x) = P_k ( x) \exp \left( - \frac{\alpha | x |}{\left(
    n - 1 + 2 k \right)}^{} \right) \label{efatom}
  \end{equation}
  for any harmonic, homogeneous polynomial $P_k$ of degree $k$.
\end{example}

\begin{example}
  Hardy's inequality: $X = \alpha \frac{x}{| x |^2}$: this is the case $p =
  2$, so that $| X |^2 $ and $\tmop{div} X$ have the same exponent. The
  singularity at 0 causes again no problems (assuming $n \geqslant 3$) and we
  get
  \[ \| \nabla \phi + X \phi \|^2 = \| \nabla \phi \|_{}^2 + \alpha (
     \alpha^{} - n + 2) \int \frac{| \phi ( x) |^2}{| x |^2} d x \]
  and
  \[ \int_{\mathbbm{R}^n} | \nabla \phi ( x) |^2 d x \geqslant \frac{\left(
     n - 2 \right)^2}{4} \int_{\mathbbm{R}^n} \frac{| \phi ( x)
     |^2}{| x |^2} d x. \]
  But this time we cannot use $\left( \ref{eigenfun} \right)$. Instead, one has
  to solve
  \[ \nabla \phi + \alpha \frac{x}{| x |^2} \phi = 0 \Rightarrow \nabla ( \log
     \phi + \alpha \log | x |) = 0, \]
  yielding
  \[ \log ( \phi | x |^{\alpha}) = C \Rightarrow \phi_0 ( x) \sim | x
     |^{- \alpha} . \]
  Indeed, there is no reasonable minimzer, though formally $\mathcal{J}_X (
  \phi_0) = 0.$ Clearly, neither $\| \nabla \phi_0 \|_{}$ nor $\| \phi_0 X
  \|_{}$ is finite.
\end{example}

\begin{remark}
  Scaling behavior. Let $\phi_{\lambda} = \lambda^{n / 2} \phi ( \lambda x)$,
  $X_{\lambda} ( x) = \lambda^{- 1} X \left( \frac{x}{\lambda} \right)$
  then

  \[ \mathcal{J}_X ( \phi_{\lambda}) = \int_{\mathbbm{R}^n} \left| \lambda^{}
     \nabla \phi ( y) + \phi ( y) X \left( \frac{y}{\lambda} \right) \right|^2
     d y \]

  thus
  \[ \mathcal{J}_X ( \phi_{\lambda}) = \lambda^2 \mathcal{J}_{X_{\lambda}} (
     \phi) . \]
  The example $X_{\lambda} = \lambda^{p - 2} x | x |^{- p}$
  shows clearly that for $p \geqslant 2$ things are going odd. The examples
  also show that the singularities of $X$ decrease the regularity of the
  ground state. It is possible, of course, to extend Proposition (2) to
  punctured domains, so that the results may be extended beyond $p = 2$. 
\end{remark}

The reader may be puzzled by the form of the eigenfunctions $\left(
\ref{efosc} \right)$ and $\left( \ref{efatom} \right)$, but this is easily
resolved when recollecting the fact that a homogeneous polynomial may look
quite differently if restricted to a sphere (the factor $x_1^2 + x_2^2 + x_3^2
\ldots = 1$ drops out). There is, by the way, an interesting connection to
the fact that the Fourier transform leaves functions of the form
\[ f ( x) = P_k ( x) \psi ( | x |) \]
invariant, in the sense that
\[ \hat{f} ( \xi) = P_k ( \xi) \widehat{\psi} ( | \xi |) . \]
This is obvious in the case $p = 0$, $\alpha = 1$, as we then have the
eigenfunctions of the Fourier transform, but not so trivial in the case $p =
1$.

\section{Regularity}

We will give here only some elementary facts and assume some smoothness of
the solutions. It is well known that if $u$ and $F$ are continuous and
\[ \nabla u ( x) = F ( x)  \hspace{2em} \tmop{in} B_r ( x_0) \]
in the distributional sense, then it also holds in the classical sense
\cite{ho1}.

\begin{proposition}
  Let $\Omega \subset \mathbbm{R}^n$ be an open set and $\phi \in C^1 (
  \Omega)$ a solution of
  \[ \nabla \phi ( x) + \phi ( x) X ( x) = 0 \hspace{2em} \forall x \in \Omega
  \]
  where $X$ is a vector field on $\Omega$. Then $X$ is locally the gradient of
  a $C^1$ function (and therefore continuous) on the open set $\{ x \in \Omega
  : \phi ( x) \neq 0 \}$. Moreover, if $X \in C^k ( \Omega)$, $k \in
  \mathbbm{N}$, then $\phi \in C^{k + 1} ( \Omega)$.
\end{proposition}

\begin{proof}
  Suppose $\phi ( x_0) > 0$ for some point $x_0 \in \Omega$, then $\phi ( x) >
  0$ in $\tmop{the} \tmop{ball} B_{\epsilon} ( x_0)$ for some $\epsilon =
  \epsilon ( x_0) > 0$ by continuity. Thus, $X ( x) = - \nabla \log \phi ( x)$
  in $B_{\epsilon} ( x_0)$. If $\phi ( x_0)<0$ then apply the same argument to
  the function $\phi' ( x) \assign - \phi ( x)$, giving $X ( x) = - \nabla
  \log \phi' ( x)$. Since $\log \phi$ is $C^1$ wherever $\phi>0$ is, the
  assertion follows. The regularity claim can be proved by applying Leibniz's
  rule. 
\end{proof}

\begin{lemma}
  Let $\phi$ be a function defined in the ball $\{ | x | < R_0 \}$ satisfying
  \[ | \phi ( x) | \leqslant C | x |^{\alpha} \sup_{| y | \leqslant | x |}  |
     \phi ( y) | \]
  for some constants $C \geqslant 0$, $\alpha > 0$. Then $\phi \equiv 0$ in
  the ball $\{ | x | < \min ( R_0 \nocomma, C^{- \alpha}) \}$. \ 
\end{lemma}

\begin{proof}
  Let $M ( r) : = \sup_{| y | \leqslant r}  | \phi ( y) |$ and note that it is
  a non-decreasing function of $r$. Then we have by assumption $| \phi ( x) |
  \leqslant C | x |^{\alpha} M ( | x |)$, therefore, taking the supremum:
  \[ \sup_{| x | \leqslant r}  | \phi ( x) | = M ( r) \leqslant C r^{\alpha}
     \sup_{| x | \leqslant r} M ( | x |) = C r^{\alpha} M ( r) \]
  thus, $M ( r) = 0$ for $C r^{\alpha} < 1$. \ 
\end{proof}

\begin{proposition}
  Let $\phi \in C^1 ( \mathbbm{R}^n)$ be a solution of
  \[ \nabla \phi ( x) + \phi ( x) X ( x) = 0 \hspace{2em} \forall x \in
     \mathbbm{R}^n \]
  where $X$ is a vector field on $\mathbbm{R}^n$. Suppose $\phi ( x_0) = 0$
  for a $x_0 \in \mathbbm{R}^n .$ If $X$ is locally bounded then $\phi \equiv
  0$ in $\mathbbm{R}^n .$ \ 
\end{proposition}

\begin{proof}
  Without loss of generality let $x_0 = 0$. Since $\phi \in C^1 (
  \mathbbm{R}^n)$ it follows that $\frac{d}{d t} \phi ( t x) = \langle \nabla
  \phi ( t x), x \rangle = -\langle \phi ( t x) X ( t x), x \rangle \Rightarrow
  \phi ( x) - \phi ( 0) = -\int_0^1 \langle \phi ( t x) X ( t x), x
  \rangle d t$. Thus
  \[ | \phi ( x) | = \left| \int_0^1 \langle \phi ( t x) X ( t x), x \rangle d
     t \right| \leqslant | x | \sup_{| y | \leqslant | x |}  | \phi ( y) X (
     y) | . \]
  
\end{proof}

\begin{remark}
  Local boundedness is necessary as the standard $C_0^{\infty}$ function
  $\chi^{}_{B_1} e^{\frac{1}{| x |^2 - 1}}$ shows.
\end{remark}

\section{Concluding remarks}

We have seen that the Schr{\"o}dinger ground states and those of the function
$\mathcal{J}_X$ are essentially the same. Moreover, ground states cannot
change sign and the vector field $X$ has (therefore) to be a gradient. Thus,
they have to satisfy the simple first order equation
\[ \nabla \phi ( x) + \phi ( x) \nabla u ( x) = 0 \]
with solution
\[ \phi ( x) = C e^{- u ( x)} . \]
The qualitative properties of $\phi$ are determined in an essential way by the
function $u$. For instance, the critical points of $\phi$ are those of $u$,
and $\phi$ is log-concave if $u$ is convex. In a certain way the same is true
for the ``inhomogeneous'' equation
\[ \nabla \phi ( x) + \phi ( x) X ( x) = Y ( x), \]
which we have not touched here.

Now, what is the physical meaning of $X ?$ When we multiply $\left(
\ref{maineq} \right)$ by $- i$, we get
\begin{equation}
  p \phi = - i \nabla \phi = i X \phi, \label{momentum}
\end{equation}
where $p$ is the momentum operator. So, indeed, $X$ indicates the momentum of
the ground state. In the same sense it holds for the angular momentum,
\[ L \phi = ( x \wedge p) \phi = i ( x \wedge X) \phi . \]
For any radial function $u ( x) = u ( | x |)$, for example:
\[ X = \nabla u = u' ( | x |) \frac{x}{| x |} \]
so that the angular momentum of the corrseponding ground state has to be zero.
In other words, one can prescribe the momentum 'field' $X = \nabla
u$ (necessarily a potential field), then the (Schr{\"o}dinger) ground state is
completely determined. Clearly, the function $u ( x)$ is closely related to
the classical action function $S ( x)$ and the parallels to Hamilton-Jacobi
theory are quite obvious, but we have found that the deeper reason for $\left(
\ref{momentum} \right)$ stems from the quantum mechanical phase-space measure
\[ d \mu_{\phi} = | \phi ( x) |^2  | \hat{\phi} ( k) |^2 d x d k, \]
which is a Radon measure generated by any normalized $L^2 ( \mathbbm{R}^n)$
function $\phi$, so that one can write for the quantum mechanical energy:
\[ \mathcal{E} ( \phi) = \int_{\Gamma} \mathcal{H} ( x, \hbar k) d \mu_{\phi},
\]
where $\mathcal{H}$ is the classical Hamilton function. We cannot go into
details, but minimizing $\mathcal{E}$ with respect to $\phi$, gives not the
usual Schr{\"o}dinger equation, however, a kind of ``double one'' on $L^2 (
\Gamma)$, where $\Gamma =\mathbbm{R}^n \times \mathbbm{R}^n$ is the phase
space. When we write $\left( \ref{peq2} \right)$ in the form
\[ \int_{\mathbbm{R}^n} | k |^2 | \hat{\phi} ( k) |^2 d k \int_{\mathbbm{R}^n}
   | X |^2 | \phi ( x) |^2 d x \geqslant \frac{1}{4} \left(
   \int_{\mathbbm{R}^n} \tmop{div} X | \phi ( x) |^2 d x \right)^2  \]
and rewrite it to
\[ \int_{\Gamma} | k |^2  | X |^2 d \mu_{\phi} 
   \geq \vert\langle \widehat{X \phi}, k\hat{\phi} \rangle\vert^2 \]
the symmetry between $k$ and $X$ may be apparent. It might give a clue why it 
is the
Fourier transform that connects configuration and momentum space in quantum
mechanics (the way the Heisenberg group acts on  
$\Gamma\times\mathbbm{R}$ and the behaviour of $d\mu_{\phi}$ under canonical
tranformations also support this).

Another way to gain some physical insight is to look at the associated 
energy-momentum tensors to the equations \ $\left(
\ref{schreq} \right)$ and $\left( \ref{eulereq} \right)$ :
\[ T_S = \nabla \phi \otimes \nabla \phi -\mathbbm{I} ( | \nabla \phi |^2 + (
   V ( x) - E) | \phi |^2) \]
and
\[ T_X = \nabla \phi \otimes \nabla \phi -\mathbbm{I} ( | \nabla \phi |^2 + (
   | X |^2 - \tmop{div} ( X)) | \phi |^2) . \]
Now,
\[ \tmop{Div} ( T_S) = - \nabla V ( x) | \phi |^2 \]
and (think of a force $F(x)= - \nabla V ( x)$)
\[ \tmop{Tr} ( T_S) = ( 2 - n) | \nabla \phi |^2 + n ( V ( x) - E) | \phi |^2
   . \]
When we insert $\nabla \phi = - \phi X$ into $T_X$, it follows
\[ T_X = | \phi |^2  ( X \otimes X -\mathbbm{I} ( 2 | X |^2 - \tmop{div} X)),
\]
and therefore
\[ \tmop{Tr} ( T_X) = | \phi |^2  ( ( 2 - n) | X |^2 + n \tmop{div} X) . \]
The meaning of $X$ and $\tmop{div} X$ is now quite obvious as $T_X $ and $T_S$
are factually the same. For instance, if the potential $V$ is homogeneous of
degree $k$, then (formally)
\[ \tmop{div} ( T_S x) = \tmop{Div} ( T_S) \cdot x + \tmop{Tr} ( T_S) = - x
   \cdot \nabla V ( x) | \phi |^2 + ( 2 - n) | \nabla \phi |^2 + n ( V ( x) -
   E) | \phi |^2 \]
\[ = ( 2 - n) | \nabla \phi |^2 + ( n - k) V ( x) | \phi |^2 - n E | \phi |^2,
\]
which is nothing more than the virial theorem, so that we must have
\[ \tmop{div} ( T_X x) = 
    \{ ( 2 - k) | X |^2 + ( k - n) \tmop{div} X - k E \} | \phi |^2 . \]
If, as we have seen in the examples, the ground states fall off sufficiently
rapid, then
\[ \int_{\mathbbm{R}^n} \tmop{div} ( T_S x) d x = 0, \]
hence

	\[
	     2 \int_{\mathbbm{R}^n} | \nabla \phi |^2 dx = 
			 k \int_{\mathbbm{R}^n} V(x) | \phi |^2 dx.
\]



\begin{thebibliography}{9}

\bibitem{jh1}
  Joseph Hersch,
  \emph{On the Methods of One-Dimensional Auxiliary Problems and of
	      Domain Partitioning: Their Application to Lower Bounds for the
				Eigenvalues of Schr\"odinger's Equation.}.
  Journal of Mathematics and Physics, Vol. XLIII,
  No. 1,
  March 1964.

\bibitem{ho1}
  Lars H\"ormander,
  \emph{Linear Partial Differential Operators}.
  Springer Verlag, Berlin,
  Fourth Printing, (Theorem 1.4.2)
  1976.
	
\end{thebibliography}
\end{document}